\documentclass[a4paper, 11pt]{article}

\usepackage[T1]{fontenc}
\usepackage{xcolor}
\usepackage[a-2b]{pdfx}
\usepackage{graphicx}
\usepackage{multirow}
\usepackage{mathtools}
\usepackage{amssymb,amsfonts}
\usepackage[amsmath,amsthm,thmmarks]{ntheorem}
\usepackage{mathrsfs}
\usepackage{authblk}
\usepackage[title]{appendix}
\usepackage{textcomp}
\usepackage{manyfoot}
\usepackage{booktabs}
\usepackage{algorithm}
\usepackage{algorithmicx}
\usepackage{algpseudocode}
\usepackage{listings}
\usepackage{hyperref}
\usepackage[capitalize,nameinlink]{cleveref}
\usepackage{bm}
\usepackage{xspace}
\usepackage{etoolbox}
\usepackage[shortlabels]{enumitem}
\usepackage{microtype}
\usepackage{etoolbox}
\usepackage{xr}
\usepackage{subfiles}
\usepackage{parskip}

\let\epsilon\varepsilon

\def\E{\mathop{{}\mathbb E}}
\def\F{\mathbb F}
\def\O{\mathcal O}
\def\R{\mathbb R}
\def\X{\mathcal X}
\def\Z{\mathbb Z}

\def\Ft{\F_2}

\def\Fq{\F_q}

\def\Ftx{\Ft[X]}
\def\Ftxn{\Ft[X] / (X^n - 1)}

\def\Fqx{\Fq[X]}

\def\ip{\angles}
\def\set{\braces}
\def\tensor{\otimes}
\def\newdef{\emph}

\DeclareMathOperator{\Ber}{Ber}

\def\statdist{\Delta_{\text{tv}}}

\def\Pn{P_{n}}

\def\A{\bm{A}}

\def\ee{\bm{e}}

\def\m{\bm{m}}

\def\s{\bm{s}}

\theoremstyle{definition}

\newtheorem{theorem}{Theorem}[section]

\newtheorem{lemma}{Lemma}[section]
\newtheorem{corollary}{Corollary}[section]

\newcommand\swapifbranches[3]{#1{#3}{#2}}
\makeatletter
\MHInternalSyntaxOn
\patchcmd{\DeclarePairedDelimiter}{\@ifstar}{\swapifbranches\@ifstar}{}{}
\MHInternalSyntaxOff
\makeatother

\DeclarePairedDelimiter\parens{\lparen}{\rparen}
\DeclarePairedDelimiter\abs{\lvert}{\rvert}
\DeclarePairedDelimiter\norm{\lVert}{\rVert}
\DeclarePairedDelimiter\weight{\lVert}{\rVert}

\DeclarePairedDelimiter\braces{\lbrace}{\rbrace}
\DeclarePairedDelimiter\bracks{\lbrack}{\rbrack}
\DeclarePairedDelimiter\angles{\langle}{\rangle}

\newcommand{\eqdef}{\stackrel{\textrm{def}}{=}}
\newcommand{\ie}{\textit{i.e.,}\xspace}

\newcommand{\bfa}{\mathbf{a}}

\newcommand{\bfe}{\mathbf{e}}

\newcommand{\bfs}{\mathbf{s}}

\newcommand{\bfA}{\mathbf{A}}

\newcommand{\bfM}{\mathbf{M}}

\newcommand\bmhead[1]{\textbf{#1.}}

\usepackage[
style=alphabetic,
   giveninits=true,
]{biblatex}

\DeclareSourcemap{
  \maps[datatype=bibtex]{
\map{
        \pernottype{book}
        \step[fieldsource=doi,final]
        \step[fieldset=isbn,null]
      }
\map{
         \step[fieldsource=doi,final]
         \step[fieldset=eprint,null]
       }
\map[overwrite]{
         \step[fieldsource=shortjournal]
         \step[fieldset=journaltitle,origfieldval]
      }
\map{
         \step[fieldset=issn,null]
         \step[fieldset=urldate,null]
         \step[fieldset=location,null]
         \step[fieldset=editor,null]
         \step[fieldset=series,null]
         \step[fieldset=pagetotal,null]
      }
   }
}

\bibliography{bibliography/bibliography.bib,bibliography/codecrypto.bib}

\patchcmd{\bibsetup}{\interlinepenalty=5000}{\interlinepenalty=10000}{}{}

\begin{document}

\title{On the Independence Assumption in \\ Quasi-Cyclic Code-Based
  Cryptography}\author[1]{Maxime Bombar}\author[2]{Nicolas Resch}\author[2]{Emiel Wiedijk}\affil[1]{Cryptology Group, CWI, Amsterdam, The Netherlands \protect
  \\ Institut de Mathématiques de Bordeaux, France \authorcr
  \texttt{maxime.bombar@math.u-bordeaux.fr}\vspace{\baselineskip}}
\affil[2]{University of Amsterdam, The Netherlands \authorcr \texttt{\{n.a.resch,e.wiedijk\}@uva.nl}}
\maketitle

\begin{abstract}

  Cryptography based on the presumed hardness of decoding codes --
  i.e., code-based cryptography -- has recently seen increased
  interest due to its plausible security against quantum attackers.
  Notably, of the four proposals for the NIST post-quantum
  standardization process that were advanced to their fourth round for
  further review, two were code-based. The most efficient proposals --
  including HQC and BIKE, the NIST submissions alluded to above -- in
  fact rely on the presumed hardness of decoding \emph{structured}
  codes. Of particular relevance to our work, HQC is based on
  \emph{quasi-cyclic codes}, which are codes generated by matrices
  consisting of two cyclic blocks.

  In particular, the security analysis of HQC requires a precise
  understanding of the Decryption Failure Rate (DFR), whose analysis
  relies on the following heuristic: given random ``sparse'' vectors
  $e_1,e_2$ (say, each coordinate is i.i.d. Bernoulli) multiplied by
  fixed ``sparse'' quasi-cyclic matrices $A_1,A_2$, the weight of
  resulting vector $e_1A_1+e_2A_2$ is very concentrated around its
  expectation. In the documentation, the authors model the
  distribution of $e_1A_1+e_2A_2$ as a vector with \emph{independent}
  coordinates (and correct marginal distribution). However, we uncover
  cases where this modeling fails. While this does not invalidate the
  (empirically verified) heuristic that the weight of $e_1A_1+e_2A_2$
  is concentrated, it does suggest that the behavior of the noise is a
  bit more subtle than previously predicted. Lastly, we also discuss
  implications of our result for potential worst-case to average-case
  reductions for quasi-cyclic codes.
\end{abstract}

 \section{Introduction}

In light of recent calls for post-quantum secure cryptography -- i.e., cryptography that is secure in a world with quantum computers
-- code-based cryptography has recently seen a growth in interest as a
prominent candidate for quantum-safe cryptography. In particular, all
three remaining finalists in the 4th round of NIST competition are
code-based~\cite{BIKEr4,HQCr4,ABCCGLMMMNPPPSSSTW20}. Informally,
code-based cryptographic schemes are those whose security can be
reduced to the conjectured hardness of decoding linear codes under the
Hamming metric.

More concretely, the quintessential hard problem for code-based
cryptography is the decoding problem (also sometimes called
\emph{learning parity with noise (LPN)}), which asks one to recover
$\s$ from the input $(\A,\s\A+\ee)$,\footnote{Technically, we are describing
  the search version of LPN. For cryptographic purposes a decision
  variant is often required, which states that given $\A$
  distinguishing $\s\A+\ee$ from a uniformly random vector is hard.
  However, due to a search-to-decision reduction~\cite{FS96},
they are polynomially equivalent.} where $\A \gets \F_2^{n \times \m}$
and $\s \gets \F_2^n$ are uniformly distributed, and $\ee \in \F_2^m$ is
a \emph{noise} vector where each $e_i$ is an independent Bernoulli
variable (i.e., it is $1$ with probability $p$ and $0$ with
probability $1-p$).\footnote{In fact, for technical reasons, it is
  often easier to consider LPN as an \emph{oracle} problem, as we do
  later. The complexity of the two variants are polynomially related,
so they are interchangeable for our purposes.}

This problem inspired the closely related \emph{learning with errors
(LWE)} problem, which is at the core of lattice-based
cryptography. Here, the noise $\ee$ is sampled differently,
typically as a (rounding of) a Gaussian random variable (one is also required to work over a large field).

When constructing public-key cryptography from either LPN or LWE,
the matrix $\A$ always forms (a part of) the public key. Thus, one
is required to publish at least $nk$ field elements: this quadratic
lower bound on the public-key size often renders these schemes
uncompetitive in terms of efficiency. To remedy this situation, it has been proposed to instead sample \emph{structured} matrices $\A$: for such
matrices, it is still plausible (and widely believed) that (quantum)
algorithms cannot efficiently solve the relevant decoding problem;
however, their structure allows for a much more succinct
representation, ideally with only $n$ field elements. This is
precisely the approach taken by many NIST submissions~\cite{BIKEr4,HQCr4}. In
particular, these schemes are based on \emph{quasi-cyclic
codes}, which we now introduce.

\paragraph{Quasi-cyclic codes.} A quasi-cyclic code is a code that
is generated by a matrix composed of multiple blocks of circulant
submatrices, \ie matrices such that each row is a circular shift of
its first row. Consequently, each submatrix can be represented by
storing only its first row. For instance, in quasi-cyclic codes of rate $1/2$, only $2n$ field elements need to be stored.

One important advantage of quasi-cyclic codes in cryptographic
applications is their polynomial representation. Specifically, let
\(\mathcal{R} \eqdef \Fq[X]/(X^{n}-1)\). A circulant matrix of
the form
\[
  \bfM_{a} \eqdef
  \begin{pmatrix}
    a_{0} & a_{1} & \dots & \dots & a_{n-1}\\
    a_{n-1} & a_{0} & \dots & \dots & a_{n-2} \\
    \vdots & \ddots & \ddots & & \vdots \\
    \vdots & & \ddots & \ddots & \vdots \\
    a_{1} & a_{2} & \dots & a_{n-1} & a_{0}
  \end{pmatrix} \in \Fq^{n\times n}
\]
represents the endomorphism
\(P(X) \in \mathcal{R} \mapsto a(X) \cdot P(X) \in \mathcal{R}\) in
the monomial basis, where \(a(X) = \sum_{i=0}^{n-1}a_i X^i\).

For example, an instance \((\A, \s\A + \ee)\) where \(\A\) is of the
form
\[
  \bfA \eqdef
  \begin{pmatrix}
    \bfM_{a_{1, 1}} & \dots & \bfM_{a_{1, r}} \\
    \vdots & \ddots & \vdots \\
    \bfM_{a_{\ell, 1}} & \dots & \bfM_{a_{\ell, r}}
  \end{pmatrix}
\]
can be compactly represented by a collection of \(r\) samples of the
form \((\bfa, \langle \bfs, \bfa \rangle + \bfe)\), where \(\bfa\)
is a vector of \(\ell\) polynomials in \(\mathcal{R}\), and
\(\langle \bfs, \bfa \rangle = \sum_{i=0}^{\ell-1}s_{i}(X)\cdot a_{i}(X)\).
For cryptographic applications, it is common to consider the
case where \(\ell = 1\), resulting in samples of the form
\((a, a \cdot s + e)\), where \(a \in \mathcal{R}\).

Such quasi-cyclic codes are employed in the NIST submissions
HQC~\cite{HQCr4} and BIKE~\cite{BIKEr4}, both of which have advanced
to the fourth round of the post-quantum cryptography standardization
process.

\paragraph{Analysis of noise}  In the analysis of error vector distribution of
HQC \cite[par.~2.4]{HQCr4}, one requires an understanding of the product of
polynomials $t(X)$ and $R(X)$, where $t(X) \in \Ft[X] / (X^n - 1)$ is a fixed
polynomial and $R(X)$ is a polynomial whose coefficients are independently
Bernoulli distributed. In fact, they require the analysis of two independent
copies of such products $e(X) \eqdef t_1(X) R_1(X) + t_2(X) R_2(X)$. To make the analysis tracetable, the authors make the
simplifying assumption that the coefficients of this $e(X)$ are
independent.

In this work we reconsider this assumption. To set up our result, we
quickly introduce some notation. We write $X \leftarrow \Ber(\omega)$
to denote a $\F_2$-valued random variable such that $\Pr[X=b] =
\frac{1+(-1)^b2^{-\omega}}{2}$.\footnote{Below, we justify this
parametrization for the Bernoulli random variable.}

\begin{theorem}[Main Theorem (Informal); see {\Cref{thm:noiseStatDistLower}}]
  Let $t(X) \in \Ft[X]/(X^n-1)$ be a fixed polynomial with $\tau$
  nonzero coefficients, and let $R(X) = \sum_{i=0}^{n-1}R_iX^i$ be
  such that each $R_i\leftarrow \Ber(\omega)$ independently. Let
  $I(X) = \sum_{i=0}^{n-1}I_iX^i$ where each $I_i \leftarrow
  \Ber(\tau\omega)$ independently. Assuming $\omega \geq \Omega(\log
  n)$, the statistical distance between $t(X)R(X)$ and $I(X)$ is
  $\Omega(\sqrt{n}2^{-2\omega})$.
\end{theorem}

We note that, by an application of the Piling-up lemma
(\Cref{lem:pilUp}), it follows that $t(X)R(X)$ and $I(X)$ share the
same marginal distribution for each coordinate: \ie, for each $0 \leq
i \leq n-1$ and $b \in \F_2$ we have $\Pr[(tR)_i=b] =
\frac{1+(-1)^b2^{-\omega}}2 = \Pr[I_i=b]$. Hence, the ``source'' of
the statistical distance is the lack of independence between the
coordinates of $t(X)R(X)$. Furthermore, we remark that
\cite[par~2.4]{HQCr4} considers (in our parametrization) $\omega =
\Theta(\log n)$, i.e., the lower bound is indeed met.

Now, note that this does not directly invalidate the modelling of
HQC: there, they consider two \emph{independent} polynomials
$t_1(X)R_1(X)$ and $t_2(X)R_2(X)$, and then model their \emph{sum}.
While we cannot invalidate this modelling (and in fact, we suspect it
might in most cases be valid), we do point out some cases where the
modelling fails:
\begin{itemize}
  \item Suppose $\langle t_1, t_2 \rangle \subsetneq \Ft[X]/(X^n - 1)$, that
    is the ideal generated by the noise does not span the entire space. Then
    $t_1(X)R_1(X)+t_2(X)R_2(X)$ can never be statistically close to a Cartesian
    product of Bernoulli polynomials, as they do span the entire
    space. Under some
    reasonable assumptions polynomials of odd weight are invertible,
    so in practice this case is easy to avoid.
\item As an extreme case, if $t_1(X)=t_2(X)\eqqcolon t(X)$, then
    $t_1(X)R_1(X)+t_2(X)R_2(X) = t(X)(R_1(X)+R_2(X))$, and since
    $R_1+R_2$ can again be modelled by an independent Bernoulli
    polynomial, \Cref{thm:noiseStatDistLower} applies.
  \item Suppose now that the support sets of $t_1$ and $t_2$ (i.e., the
    indices of the nonzero coefficients) are both in \emph{arithmetic
    progressions} -- i.e.,  sets of the form $\{ax+b\pmod{n}:x \in
    \{0,1,\dots,\tau-1\}\}$ -- with the same common difference $a$. Then we can
    still show a nontrivial lower bound on the statistical distance.
\end{itemize}

For less ``structured'' cases of $t_1,t_2$ (or even, say, $t_1,\dots,t_s$ with
$s \geq 2$) we conjecture that such a gap in the statistical distance does not
persist.

\paragraph{Related independence heuristic in cryptography} The
question of the independence of marginals of the distribution of a
product of polynomials is not restricted to the analysis of code-based
cryptosystems. In particular, a similar assumption has also been made
in lattice-based cryptography to analyse the growth of the noise in
the context of~\emph{fully hommomorphic encryption}
(FHE)~\cite[Assumption 3.11]{JC:CGGI20}. Nevertheless, recent works
have began to suggest that this did not actually
hold~\cite{EPRINT:BMCM23,MP24}, which led to underestimating this
noise growth. Our results align with those observation.

\subsection{Overview of Techniques}

In order to lower bound the statistical distance, we in fact find it easier to
work with the \emph{Kullback-Leibler (KL) divergence} between the two
distributions. Pinsker's inequality shows that these two quantities are
intimately related; however, as we are seeking a \emph{lower bound} on the
statistical distance, this inequality is not directly applicable for our
purposes. Fortunately, under mild ``regularity'' conditions (namely, the ratio
  between the two considered probability distributions is never too
  large nor too
small), we can apply a ``reverse'' Pinsker's inequality~\cite{binette2019}. As
these regularity conditions hold for our distributions of interest (see
\Cref{subsec:stat-dist}), we can focus on the KL divergence.

We begin with a convenient lemma that may be folklore, but for lack of a
suitable reference (and because we believe it might be of independent interest)
we provide a proof (\Cref{lem:divDependence}): namely, that if $Q$ and $P$
are distributions over an $n$-fold Cartesian product with matching marginal
distributions and $Q$ follows a product distribution (\ie its coordinates are
independent), then the KL-divergence $D(P\|Q)= H(Q)-H(P)$, where $H(\cdot)$ is
the \emph{(Shannon) entropy} of the distributions. For our purposes, $P$ denotes
the distribution of $t(X)R(X)$ and $Q$ denotes the distribution of $I(X)$, where
we identify their support $\Ft[X]/(X^n-1)$ with $\Ft^n$ via the natural
isomorphism. Hence, as $H(Q)$ is easy to compute (being a Cartesian product of
Bernoulli distributions), we focus on upper bounding $H(P)$.

Here, we can consider two cases. Firstly, if $t(X)$ happens to not be
invertible, then already $t(X)R(X) \in \langle t(X) \rangle \subsetneqq
\Ft[X]/(X^n-1)$, where $\langle t(X) \rangle = \{t(X)a(X):a(X) \in
\Ft[X]/(X^n-1)\}$ is the ideal generated by $t(X)$. That is, $P$ is distributed
over a strict subset of $\Ft[X]/(X^n-1)$ size at most $2^{n-1}$; this already
guarantees a $H(P) \leq n-1 \ll H(Q)$ for our parameters of interest.

Otherwise, $t(X)$ is invertible. Then, it naturally follows that $H(P) =
H(R(X))$, i.e., just the distribution of $R(X)$. Indeed, multiplying by $t(X)$
is then a bijection from $\Ft[X]/(X^n-1)$ to itself, so it does not affect the
entropy. And we can again easily compute $H(R(X))$: it is again a Cartesian
product of Bernoulli distributions! That is, we can conclude
\[
  H(Q)-H(P) = n(\tilde h(\tau\omega)-\tilde h(\omega)) \ ,
\]
where $\tilde h(x)$ is the entropy of a $\Ber(x)$ random variable, and we recall
$\tau$ is the number of nonzero coefficients of $t(X)$. To conclude our desired
theorem, it suffices to lower bound $\tilde h(\tau\omega)-\tilde h(\omega)$,
which we do by expanding the Taylor series representation of $\tilde h$.

\medskip

Next, we consider cases where we can understand the entropy of $t_1R_1+\cdots+t_sR_s$, where $t_1,\dots,t_s$ are fixed polynomials and $R_1,\dots,R_s$ are independent Bernoulli polynomials (i.e., their coefficients are sampled independently). To make progress in this case, we write $t_1R_1+\cdots+t_sR_s = \sum_i C_i X^i$ and bound the entropy via a sum over roughly $n/2$ pairwise entropies $H(C_i,C_j)$. The formula for the joint distribution of $C_i$ and $C_j$ is not too difficult to obtain (and in fact has been obtained by prior work~\cite{pacher2016}), and one can observe that the joint entropies $H(C_i,C_j)$ are small if for many of the $t_\ell$'s, many of its nonzero coefficients overlap with many nonzero coefficients of $X^{j-i}t_\ell$. If $t_1,\dots,t_s$ are all of the form $\sum_{i=0}^{\tau-1}X^{a\cdot i + b}$ where $a \in \Z_n^*$ and $b \in \Z_n$ (i.e., the nonzero coefficients form an \emph{arithmetic progression}) we can show that the entropy bound will indeed be quite small.

\subsection{Future Directions}

We conclude the introduction with some directions that we leave open for future work.

\paragraph{Concentration of noise weight.} In this work, we provided an analysis
of $t(X)R(X)$ -- the product of a fixed polynomial and an i.i.d. Bernoulli
polynomial -- and showed that it is ``far'' from the distribution of an i.i.d.
Bernoulli polynomial. As discussed above, this has implications for
code-based cryptosystems such as HQC, where in order to allow for successful
decoding it is important that the weight of the noise $t_1(X)R_1(X) + t_2(X)R_2(X)$ be tightly
concentrated around its expected value. While such a concentration naturally
follows if the coordinates were indeed independent, but as we showed here in some cases that does not hold. However, this does not itself disprove the assertion that
the weight is concentrated, and indeed empirical evidence suggests that the
weight is sufficiently concentrated. Additionally, prior work~\cite{kawachi24}
already gave some concentration bounds (in this case, via Chebyshev's
inequality). We leave it as an open problem to provide further theoretical
evidence for the concentration of this weight.

\paragraph{Potential for worst-case to average-case reduction?} For
cryptographic purposes, it is of course vital that the hardness assumptions hold
for \emph{average-case problems}: namely, it is hard to solve some computational
problem (such as the decoding problem) when the instances are sampled randomly.
However, from a complexity-theoretic perspective we have a much firmer theory of
the hardness of \emph{worst-case problems}. That is, we have a more mature theory of
which problems are hard when the instances for a given algorithm are chosen
\emph{adversarially}.

In the case of LWE, Regev~\cite{regev2005,brakerski2013} famously showed that
the average-case LWE problem can be reduced to certain worst-case problems on
lattices. Inspired by this, Brakerski et al.~\cite{brakerski2019} recently
introduced a worst-case to average-case reduction for codes: namely, a reduction
from the classic worst-case decoding problem where $\A$, $\s$ and $\ee$ are
adversarially chosen, with the promise that $\ee$ has Hamming weight at most
$\tau$. One can imagine generalizing this to the quasi-cyclic case, as was
successfully done in the case of lattices (the analogous problem there is
typically termed Ring-LWE). Here, if one works with a rate $1/s$ quasi-cyclic
code the natural reduction strategy takes as input a noisy codeword
$(a_1(X)m(X)+t_1(X),\dots,a_s(X)m(X)+t_s(X))$ (with the sum of the weights of
the noise vectors $t_i$ being at most $\tau$), and then produces ``Ring-LPN''
like samples by sampling a ``smoothing'' vector $(R_1(X),\dots,R_s(X))$ and
considering
\[
  \left(\sum_{i=1}^s a_i(X)R_i(X), \left(\sum_{i=1}^s
  a_i(X)R_i(X)\right)m(X) + \sum_{i=1}^s t_i(X)R_i(X)\right) \ .
\] At the very least, this requires us to
analyze the distribution $\sum_{i=1}^s t_i(X)R_i(X)$, as we undertake in this
work.\footnote{More precisely, we make a step assuming one is choosing the
Bernoulli distribution to smooth. Other choices could be made, but we view this
as a natural first step.} Furthermore, for the \emph{standard} Ring-LPN
assumption, one must have $\sum_{i=1}^s t_i(X)R_i(X)$ close to an independent
Bernoulli polynomial. Our work points out that at least some structural
assumptions must be made on the vectors $t_1,\dots,t_s$: for example, if they
are all equal, then this reduction is doomed to fail as $\sum_{i=1}^s
t_i(X)R_i(X)$ is necessarily far from an independent Bernoulli polynomial.
Furthermore, if each vector $t_1, \ldots, t_s$ form an \emph{arithmetic
progression} with the same common difference, we can also show that
$\sum_{i=1}^s t_i(X)R_i(X)$ is far from an independent Bernoulli polynomial.

\paragraph{When is statistical distance small?} We suspect
that, given appropriate assumptions, $\sum_{i=1}^s t_i(X)R_i(X)$ is
statistically close to an independent Bernoulli polynomial. Recall
these assumptions are equivalent to the conditions such that the
entropy \[ H\parens{\sum_{i=1}^s t_i(X)R_i(X)}
\] is sufficiently high. While we have found conditions on the noise that ensure
that this entropy is low, we currently do not know any conditions that yield a
high entropy. In general, computing the entropy of the sum independent random
variables is hard, refer for example to \cite{tao2010, green2024}. To the best
of our knowledge, computing the entropy of this polynomial in general is an open
problem.

The conditions for which the entropy $H(\sum_{i=1}^s t_i(X)R_i(X))$ is known to be small
have an important caveat: the relevant worst case decoding problem \[
  (a_1(X) m(X) + t_1(X), \ldots, a_s(X) m(X) + t_s(X))
\]  is in fact \emph{easy}. If all noise vector vectors are identical,
then we can decode $(a_1(X) m(X) + t(X), a_2(X) m(X) + t(X))$. By computing the
difference $((a_1(X) + a_2(X)) m(X))$ it is easy to decode to $m(X)$, if $a_1(X)
+ a_2(X)$ is invertible. Decoding is also easy when the noise is guaranteed to
be an arithmetic progression. As the number of arithmetic progressing noise
vectors is polynomial (when $s$ is fixed), it is easy to brute force all
possible noise vectors. We remain especially interested in finding the entropy
$H(\sum_{i=1}^s t_i(X)R_i(X))$ for cases where the worst-case problem is assumed
to be hard.

 \section{Preliminaries}
\bmhead{General notation} For positive integer $n$, we write
$[n] = \{1,2,\dots,n\}$.

We choose a (somewhat) nonstandard definition for the Bernoulli random
variable: for $\omega \in [0,+\infty]$, we say $x \leftarrow \Ber(\omega)$
if \(x\in \Ft\) and
\[
  \Pr[x=b] =
  \begin{cases}
    \frac{1-2^{-\omega}}{2} & b=1\\
    \frac{1+2^{-\omega}}{2} & b=0
  \end{cases}
\]
In other words, $\omega$ is the log of the bias of the Bernoulli. For
positive integer $n$ we let $\Ber(\omega)^{\otimes n}$ to denote a
vector $(x_1,\dots,x_n)$, where each $x_i \leftarrow \Ber(\omega)$
independently. We sometimes abuse notation and write
$R(X) \leftarrow \Ber(\omega)^{\otimes n}$ to mean that
$R(X) = \sum_{i=0}^{n-1} R_i X^i$ and each $R_i \leftarrow
\Ber(\omega)$. We choose this parametrization as the statement of the
\emph{piling-up lemma} -- which determines the distribution of the
sum modulo 2 of Bernoulli random variables -- becomes very simple.

\begin{lemma}[Piling-up lemma] \label{lem:pilUp} Let
  $X \leftarrow \Ber(\omega_x)$ and $Y \leftarrow \Ber(\omega_y)$ be independent
  random variables. Then $X + Y \leftarrow \Ber(\omega_x + \omega_y)$.
\end{lemma}

For distributions $P$ and $Q$ over $\Omega$, we use the following
notation for information-theoretic quantities:
\begin{align*}
  H(P)      & \eqdef \sum_{x \in \Omega} P(x) \log(1/P(x)) &&
  \text{(entropy)} \\
  D(P \| Q) & \eqdef \sum_{x \in \omega}P(x) \log(P(x) / Q(x)) &&
  \text{(Kullback-Leibler divergence)} \\
  \statdist(P, Q) & \eqdef \frac{1}{2} \sum_{x \in \omega} \abs{P(x) - Q(x)}
  && \text{(statistical distance)}
\end{align*}
For convenience, we abuse notation by writing random variables instead of the
distribution of these random variables. For example, when $X
\leftarrow P$, we write
$H(X)$ instead of $H(P)$.

For the binary entropy we write $h(x) \eqdef -x \log(x) - (1 -
x)\log(1-x)$ for the
binary entropy. We additionally write $\tilde h(\omega)$ for the
entropy of a $\Ber(\omega)$ random variable, so
\[
  \tilde h(\omega) \eqdef h\parens{\frac{1 - 2^{-\omega}}{2}} \ .
\]

Furthermore, we write \[
  p(\omega) \eqdef \frac{1 - 2^{-\omega}}{2},
\]
for the probability of sampling $1$ in the distribution $\Ber(\omega)$.
 \section{Analysis}
The general goal of this article is to analyze the distribution of $t(X) R(X)$.
Here $t(X)$ is some fixed polynomial in $\Ft[X] / (X^n - 1)$, and $R(X)
\leftarrow \Ber(\omega)^{\tensor n}$. Specifically, we want to know when the
coefficients of this polynomial are close to independent. First, we give the
marginal distribution of these coefficients, which is certainly
folklore, but we state as a lemma for future convenience.

\begin{lemma} \label{lem:rtBerSum}
  Let $\Pn := \Ft[X] / (X^n - 1)$ be the polynomial quotient ring. Let $t
  \in \Pn$ be a fixed polynomial, let $R \leftarrow
  \Ber(\omega)^{\tensor n}$ be a
  random variable.

  Write $(t R)_i$ for the coefficient before $X^i$ so
  \[
    t R = \sum_{i=0}^{n - 1} (t R)_i X^i.
  \] Then for all $k$ in $\{0, \ldots, n - 1\}$  \[
    (t R)_k = \sum_{\substack{j \in \{0, \ldots, n - 1\} \\ t_{k - j}
    = 1}}^{n - 1} R_j,
  \] where for convenience, computations in the indices are modulo $n$.
\end{lemma}
\begin{proof}
  The lemma follows from a simple rewriting of the polynomial

  \begin{align*}
    t R
    & = \sum_{i=0}^{n - 1} \sum_{j=0}^{n - 1} t_i R_j X^{i + j} \\
    & = \sum_{k=0}^{n - 1} (\sum_{j = 0}^{n - 1} t_{k - j} R_j)X^{k}
    && \text{(let $k = i + j$)} \\
    & = \sum_{k=0}^{n - 1} \parens{
      \sum_{\substack{j \in \{0, \ldots, n - 1\} \\ t_{k - j} = 1}}^{n - 1} R_j
    } X^{k} .
  \end{align*}
\end{proof}

From this lemma we can easily compute the marginal distributions of each
coefficient.

\begin{lemma} \label{lem:rtMarDist}
  Let $R$ and $t$ as in \cref{lem:rtBerSum}. Then \[
    (t R)_{k} \leftarrow \Ber(\weight{t} \omega)
  \]
\end{lemma}
\begin{proof}
  From \cref{lem:rtBerSum} it follows for all $k \in \{0, \ldots, n - 1\}$
  that \[ (t R)_k = \sum_{\substack{j \in \{0, \ldots, n - 1\} \\ t_{k
  - j} = 1}}^{n - 1} R_j. \] Therefore, $(t R)_k$ is the sum mod 2 of
  $\weight{t}$ independent Bernoulli variables. The piling-up lemma
  (\cref{lem:pilUp}) gives us that $(t R)_k \leftarrow \Ber(\weight{t}
  \omega)$.
\end{proof}

Now \cref{lem:rtMarDist} perfectly characterizes the marginal distribution
of the coefficients $R t$. However, this lemma does not imply that $R
t \leftarrow
\Ber(\weight{t} \omega)^{\otimes n}$. This statement would hold if all the
coefficients of $R t$ were independent. Unfortunately, the coefficients are
not independent. Different coefficients of $R t$ depend on the same
coefficients in $R$. Specifically the intersection between \[
  \set{j \in \{0, \ldots, n - 1\} \mid t_{k - j} = 1} \cap \set{j \in \{0,
  \ldots, n - 1\} \mid t_{k' - j} = 1}
\]
coefficients in $R$ may very well be non-empty creating a dependence between
$(R t)_k$ and $(R t)_{k'}$.

One property we can immediately infer from the marginal distribution is the
expectation.

\begin{lemma}[Expectation of $|t R|$] \label{lem:exptR}
  Let $t, R \in \Ftxn$ with $t$ fixed and $R
  \leftarrow~\Ber(\omega)^{\otimes n}$, then \[
    \E[|t R|] = n \cdot \frac{1 - 2^{-\weight{t} \omega}}{2}
  \]
\end{lemma}
\begin{proof}
  The weight $|t R|$ is the sum of the coefficients of $t R$. Then we can
  use the linearity of the expectation to compute the sum of the
  expectations in $\Z$. Note the addition here is defined over $\R$,
  different from the previous computations where addition was defined over
  $\F_2$. Because every $(t R)_i \leftarrow \Ber(\weight{t} \omega)$
  \begin{align*}
    \E[|t R|]
    & = \E\bracks{\sum_{i=0}^{n - 1} (t R)_i} \\
    & = \sum_{i=0}^{n - 1} \E\bracks{(t R)_i} \\
    & = n \cdot \frac{1 - 2^{-\weight{t}\omega}}{2}.
  \end{align*}
\end{proof}

Because the coefficients are dependent, many common methods to analyze the
probability that the weight $|t R|$ is close to the expectation do not apply.
In \cite{kawachi24} an analysis is provided showing that the $|tR|$
is indeed somewhat concentrated around its mean (essentially by
analyzing the variance of $|tR|$ and then applying Chebyshev's inequality).

However, this dependence between the coefficients $(t R)_i$ turns out to be
significant, at least in the sense that the statistical distance
between $tR$ and a Bernoulli distribution will be non-negligible.
First let us give some necessary conditions for $t R$ to look
like a Bernoulli distribution. A Bernoulli distribution will reach every
polynomial with non-zero probability. Specifically there is a non-zero
probability that $(t R)(X) = 1$, so $t$ must be invertible in $P_n$.

So it is necessary for $t(X)$ to be invertible, if $t(X) R(X)$ should look
like a Bernoulli distribution. If $R(X)$ is unbiased enough, then $t(X)
R(X)$ will look like a Bernoulli distribution. In the extreme case: if
$R(X)$ is uniform over $\Ftx / (X^n - 1)$, then $t(X) R(X)$ will also be
uniform. In \Cref{thm:deptR} below, we discuss how low the bias of $R(X)$ can
be for $t(X) R(X)$ to look like a Bernoulli distribution.

In order to analyze the statistical distance between $t(X)R(X)$ and a
independent Bernoulli polynomial, as mentioned in the introduction we
prefer to analyze the KL-divergence, which the following lemma states
has a relatively simple form. This result is quite likely folklore,
but for lack of a good reference, we provide a proof.

\begin{lemma} \label{lem:divDependence}
  Let $P: \X \to [0, 1]$ be a discrete distribution over $\X = \X_1 \times
  \ldots \times \X_n$, with $P_1, \ldots P_n$ the marginal distributions.
  Define \[
    Q \eqdef P_1 \otimes \ldots \otimes P_n
  \] the distribution, such that the
  marginal distributions of $Q$ and $P$ agree (so $Q_i = P_i$), and the
  marginal distributions of $Q$ are independent.
  Then \[ D(P \| Q) = H(Q) - H(P). \]
\end{lemma}
\begin{proof}
  It is a well-known fact that we can write the KL-divergence as
  \begin{align*}
    D(P\|Q)
    & =\sum_{x \in \X} P(x) \log\parens{P(x) / Q(x)}
    \\ & = \sum_{x \in \X} P(x) \log\parens{P(x)} - \sum_{x \in \X}
    P(x) \log(Q(x))
    \\ & = H(P:Q) - H(P).
  \end{align*}
  Where $H(P : Q)$ is defined to be the \newdef{cross-entropy} $-
  \sum_{x \in \X} P(x) \log(Q(x))$. We only need to show the equality
  $H(P:Q) = H(Q)$. To demonstrate this, it is convenient to introduce
  a random vector $X = (X_1,\dots,X_n)$ such that $\Pr[X=x] = P(x)$
  for all $x \in \X$. We also write $\X_{\neq k} := \X_1 \times
  \cdots \times \X_{k-1} \times \X_{k+1} \times \cdots \times \X_n$,
  i.e., it is the cartesian product of all the $\X_i$'s \emph{except}
  $\X_k$. We then define the notation
  \[
    P_{\neq k|k}(x_{\neq k}|x_k) := \Pr[X_{\neq k} = x_{\neq k}|X_k=x_k]
  \]
  where $X_{\neq k} = (X_1,\dots,X_{k-1},X_{k+1},\dots,X_n)$,
  $x_{\neq k} \in \X_{\neq k}$ and $x_k \in \X_k$.

  We can now derive
{\allowdisplaybreaks
    \begin{align*}
      & H(P : Q)
      \\ ={}& \sum_{x \in \X} P(x_1, \ldots, x_n) \log(1/Q(x_1, \ldots x_n))
      \\ ={}& \sum_{x \in \X} P(x_1, \ldots, x_n) \sum_{k=1}^n
      \log(1/Q_k(x_k)) && \text{(as $Q_1, \ldots Q_n$ independent)}
      \\ ={}& \sum_{k=1}^n \sum_{x_k \in \X_k} P_k(x_k)
      \log(1/Q_k(x_k)) \sum_{x_{\neq k} \in \X_{\neq k}} P_{\neq
      k|k}(x_{\neq k} \mid x_k) && \text{(def. conditional prob.)}
      \\ ={}& \sum_{k=1}^n \sum_{x_k \in \X} P_k(x_k)
      \log(1/Q_k(x_k)) && \text{(sum prob. is 1)}
      \\ ={}& \sum_{k=1}^n \sum_{x_k \in \X} Q_k(x_k)
      \log(1/Q_k(x_k)) && \text{(as $P_k = Q_k$)}
      \\ ={}& \sum_{k=1}^n H(Q_k) && \text{(def. entropy)}
      \\ ={}& H(Q) && \text{(as $Q_1, \ldots Q_n$ independent)}
    \end{align*}
  }
\end{proof}

The following theorem is now a simple consequence of the above lemma.

\begin{theorem}[Dependence of $t R$] \label{thm:deptR}
  Let $t(X), R(X) \in \Ftx / (X^n - 1)$ with an invertible $t$ fixed
  and $R \leftarrow
  \Ber(\omega)^{\otimes n}$. Define $I \in \Fqx / (X^n - 1)$ with $I \leftarrow
  \Ber(\weight{t} \omega)$. In other words the coefficients of $I$
  are independent,
  and have the same marginal distribution as the coefficients of~$(t
  R)(X)$. Then \[
    D(t R \| I) = n \parens{
      \tilde h(\weight{t} \omega) - \tilde h(\omega)
    }.
  \]
\end{theorem}
\begin{proof}
  This statement is a simple consequence of \cref{lem:divDependence}. The
  entropy $H(\Ber(\omega)) = \tilde h(\omega)$. So, the entropy of
  copies is $H(I) = n \tilde h(\weight{t} \omega)$. As $t$ is invertible
  we now that $H(t R) = H(R) = n \tilde h(\omega)$. In total
  \[
    D(t R \| I) = H(I) - H(t R).
  \]
\end{proof}

This theorem gives an exact analysis of the KL-divergence between $t R$ and
$n$ independent Bernoulli variables. Recall that for the purpose of our
reduction $t R$ needs to look like a Bernoulli distribution for some
parameters. So the main question is: what parameters can we pick such that
this divergence $D(t R \| I)$ is negligible?

As a sanity-check, we discuss the application of this theorem to some simple
cases. Suppose $\omega \to \infty$, then the distribution $R$ will converge
to the uniform distribution. Furthermore, $t R$ will also be uniform, so all
the coefficients will be independent. In the theorem we will have $\tilde
h(\omega) \approx \tilde h(\weight{t} \omega) \approx 1$. So \[
  D(t R \| I) = n(\tilde h(\weight{t} \omega) - \tilde h(\omega)) \approx 0.
\]

Another extreme case is $\weight{t} = 1$, so $t = X^k$ for some $k
\in \{0, \ldots,
n - 1\}$. Then multiplying by $t$ would be equivalent to shifting the
coefficients. The shift of independent Bernoulli variables still results in
independent Bernoulli variables. As $R \leftarrow \Ber(\omega)$ then $t R
\leftarrow(\Ber(\omega))$ with all the $(t R)_i$ coefficients completely
independent. In the theorem $D(t R \| I) = n (\tilde h(1 \omega) - \tilde
h(\omega)) = 0$.

When $\weight{t}$ is small or $\omega$ is big, the approximation $R t
\leftarrow \Ber(\weight{t}
\omega)$ is reasonable. On the other hand when $\omega$ is small but
$\weight{t}$
is quite big then $\tilde h(\weight{t} \omega) - \tilde h(\omega)$
becomes big. This
case is exactly when there is not enough entropy in $R$ for all the $(t
R)_i$ coefficients to be independent.

In practice, the noise is the sum of $s$ copies of $t R$, so we need to
analyze $s$ copies of this product \[
  t_1 R_1 + \ldots + t_s R_s,
\] with $R_1, \ldots, R_s \leftarrow \Ber(\omega)^{n}$ independently. Here the
total error weight relevant for the decoding problem is \[
  \tau = \weight{t_1} + \ldots + \weight{t_s}.
\] For simplicity, we will first focus on the analysis of one copy, $D(t R
\| I)$.

\subsection{Approximation for divergence of one product}
To give a stricter bound on when the divergence $D(t R \| I)$ is
small we need a lemma to approximate $\tilde h(\omega)$. This approximation
follows from a standard use of Taylor's theorem and may be folklore, but for
lack of a suitable citation we provide a proof.

\begin{lemma}[Approximation of $\tilde h(\omega)$] \label{lem:approxh}
  Let $\omega > 0$ we have that \[
    \tilde h(\omega) = 1 - \frac{2^{-2 \omega}}{\ln(2)} + \O(2^{-4\omega}).
  \]
\end{lemma}
\begin{proof}
  The proof follows from expanding the definition of $\tilde h$, and
  performing a Taylor expansion on the logarithm.
  \begin{align*}
    \tilde h(\omega)
    & = h\parens{\frac{1 - 2^{-\omega}}{2}} \\
    & = -\frac{1 - 2^{-\omega}}{2} \log\parens{\frac{1 -
    2^{-\omega}}{2}} - \frac{1 + 2^{-\omega}}{2} \log\parens{\frac{1
    + 2^{-\omega}}{2}} \\
    & = -\frac{1 - 2^{-\omega}}{2} (\log(1 - 2^{-\omega}) - 1) -
    \frac{1 + 2^{-\omega}}{2} (\log(1 + 2^{-\omega}) - 1) \\
    & = 1 - \frac{1}{2} \bigg((1 - 2^{-\omega}) \log(1 - 2^{-\omega})
    - (1 + 2^{-\omega}) \log(1 + 2^{-\omega})\bigg).
  \end{align*}

  We compute the Taylor expansion of $(1 \pm x) \log(1 \pm x)$. Because
  of the convention of using $\log$ base 2, we get an additional factor of
  $1 / \ln(2)$ in front of the usual Taylor series of the natural
  logarithm.
  \begin{align*}
    (1 + x) \log(1 + x) & = \frac{1}{\ln(2)} \parens{+x +
    \frac{x^2}{2} - \frac{x^3}{6}} + \O(x^4) \\
    (1 - x) \log(1 - x) & = \frac{1}{\ln(2)} \parens{-x +
    \frac{x^2}{2} + \frac{x^3}{6}} + \O(x^4)
  \end{align*}

  Rather than just the asymptotic behavior, we would also like to get
  an explicit lower bound on $\tilde h(\omega)$. Using Taylor's
  theorem we can compute an explicit
  formula for approximation error.

  Note that the fourth derivatives of $(1 + x)\log(1 + x)$ and $(1 - x)
  \log(1 - x)$ are
  \begin{align*}
    \frac{d^4}{d x^4} (1 + x) \log(1 + x) & = \frac{2}{\ln(2)(x + 1)^3}\\
    \frac{d^4}{d x^4} (1 - x) \log(1 - x) & = -\frac{2}{\ln(2)(x - 1)^3}
  \end{align*}
  Filling in Taylor's theorem gives that for some $\xi_{+}, \xi_{-} \in
  [0, 2^{-\omega}]$, the errors are of form
  \begin{align*}
    \epsilon_{+}
    & := (1 + x) \log(1 + x) - \frac{1}{\ln(2)} \parens{+x +
    \frac{x^2}{2} - \frac{x^3}{6}} \\
    & = \frac{\xi_{+}^4}{12 \ln(2) (1 + \xi_{+})^3} = \O(2^{-4 \omega}) \\
    \epsilon_{-}
    & := (1 - x) \log(1 - x) - \frac{1}{\ln(2)} \parens{-x +
    \frac{x^2}{2} + \frac{x^3}{6}} \\
    & = -\frac{\xi_{-}^4}{12 \ln(2) (1 - \xi_{-})^3} = \O(2^{-4 \omega}).
  \end{align*}
  Filling in the Taylor expansion yields
  \begin{align*}
    \tilde h(\omega)
    ={} & 1 - \frac{1}{2 \ln(2)}
    \begin{aligned}[t]
      \bigg( &
        -2^{-\omega}
        + \frac{2^{-2\omega}}{2}
        + \frac{2^{-3\omega}}{6}
        + \ln(2) \epsilon_{-}
        \\ &
        + 2^{-\omega}
        + \frac{2^{-2\omega}}{2}
        - \frac{2^{-3\omega}}{6}
        + \ln(2) \epsilon_{+}
      \bigg)
    \end{aligned} \\
    ={} & 1 - \frac{2^{-2\omega}}{2 \ln{2}} + \epsilon_{+} / 2 +
    \epsilon_{-} / 2.
  \end{align*}
  Using $\epsilon_{+}, \epsilon_{-} \in \O(2^{-4 \omega})$, we can
  immediately conclude \[
    \tilde h(\omega) = 1 - \frac{2^{-2 \omega}}{2 \ln(2)} + \O(2^{-4 \omega})
  \] as required.
\end{proof}

Filling in the approximation tells us when $D(t R \| I)$ is negligible.

\begin{corollary} \label{cor:tRNegl}
  Let $t, R, I \in \Ftxn$, with $t$ fixed and invertible and $R \leftarrow
  \Ber(\omega)^{\otimes n}$, and
  $I~\leftarrow~\Ber(\weight{t} \omega)^{\otimes n}$. Then \[
    D(t R \| I)
    = n \Theta(2^{-2 \omega}).
  \]
\end{corollary}
\begin{proof}
  From \cref{thm:deptR} we get $D(t R \| I) = n (\tilde h(\weight{t} \omega)
  - \tilde h(\omega))$. Filling in \cref{lem:approxh} gives
  \begin{align*}
    D(t R \| I)
    & = n\, (\tilde h(\weight{t} \omega) - \tilde h(\omega)) \\
    & = n\, \parens{1 - \frac{2^{-2\weight{t}\omega}}{2 \ln(2)} +
      \O(2^{-4\omega}) - (1 - \frac{2^{-2 \omega}}{2 \ln(2)} +
    \O(2^{-4\omega}))}\\
    & = n\cdot\Theta(2^{-2 \omega})
  \end{align*}
  as required.
\end{proof}

When $\omega$ is small \cref{cor:tRNegl} does not tell us much about the
quantity. However, even assuming that $\tilde h(\norm{t} \omega) - \tilde
h(\omega)$ remains constant, $D(t R \| I) \to \infty$ when $n \to \infty$.
Notably, if we want $D(t R \| I) \to 0$, then we need $\omega \to \infty$.

This theorem is also relevant for the post-quantum scheme HQC.
\cite{HQCr4} In Proposition 2.4.2, the noise vector $e'$ has a very
similar structure as we analyzed here, of a product of a polynomial with a
Bernoulli distribution with some other polynomial. In the paper $e'$ is
analyzed as a vector of independent entries. \Cref{thm:deptR,cor:tRNegl}
seem to suggest that making this independence assumption is too optimistic.
On the other hand the required properties which should have followed from
independence, can still hold. For example, for partial analysis on the
Hamming weight of $e'$, refer to \cite{kawachi24}.

\subsection{Statistical distance} \label{subsec:stat-dist}

\Cref{thm:deptR,cor:tRNegl} tell us when the Kullback-Leibler
divergence is small. Still, we would also like to give a bound on the
statistical distance. To get an upper bound on the statistical
distance, we can use the well-known Pinsker's inequality~\cite[Sec.
6]{kemperman1969}.

\begin{theorem}[Pinsker's inequality]
  For $P$ and $Q$ distributions \[
    \statdist(P, Q) \leq \sqrt{\frac{1}{2} D(P \| Q)}.
  \]
\end{theorem}

Pinsker's inequality immediately tells us that \[ \statdist(t R, I) \leq
  \sqrt{\frac{1}{2}\,D(t R \| I)} = \sqrt{\frac{n}{2}\cdot\parens{\tilde
h(\weight{t} \omega) - \tilde h(\omega)}}, \] and alternatively \[
\statdist(t R, I) \leq \sqrt{n} \cdot \O(2^{-\omega}). \]

Still we would also like to get a lower bound on the statistical distance.
This other direction is harder, because given a certain statistical distance
$\delta = \statdist(P, Q)$ the Kullback-Leibler divergence $D(P \| Q)$ can
be infinitely large, specifically when there is an $x$ such that $P(x) > 0$,
and $Q(x) = 0$. As $P(x)$ can be arbitrarily small, a lower bound on the
statistical distance based on the Kullback-Leibler divergence is not
possible in general.

In our case, we are working with Bernoulli variables that span the entire
outcome space $\Ftx / (X^n - 1)$. Therefore, we know that this extreme case
cannot occur. Therefore, we can use the results
in \cite{binette2019} to get a range on the statistical distance. The paper
was about arbitrary random variables. Because we just deal with
discrete random variables, we produce their results for the discrete case.

\begin{theorem}[Reverse Pinsker Inequalities \cite{binette2019}]
  \label{thm:RevPinsker} Let $P, Q: \mathcal X \to [0, 1]$ be distinct
  discrete probability distributions over $\mathcal X$. Let \[
    m = \min_{x \in X} \frac{P(x)}{Q(x)}, \,
    M = \max_{x \in X} \frac{P(x)}{Q(x)}.
  \]
  Then \[
    D(P \| Q) \leq \statdist(P, Q) \, \frac{m \log(m)}{1 - m} +
    \frac{M \log(M)}{M - 1}.
  \]
\end{theorem}
\begin{proof}
Filling in $f(x) = x \log(x)$ in \cite[Th. 1]{binette2019} gives \[
    D(P \| Q) \leq \statdist(P, Q) \,
    \parens{
      \frac{m \log(m)}{1 - m} + \frac{M \log(M)}{M - 1}
    }.
  \]
\end{proof}
Now we can apply this bound to $D(t R \| I)$.

\begin{theorem} \label{thm:noiseStatDistLower}
  For $n \geq 3$ and $\omega \geq \log(n)$, the statistical distance has a
  lower bound \[
    \statdist(t R \| I) \geq
    1 / 3 \cdot D(t R \| I).
  \]
\end{theorem}
\Cref{thm:noiseStatDistLower} is a special case of
\cref{thm:noiseStatDistLowerGeneric} where $s = 1$. We defer the proof to the
generic theorem.

The constant $1/3$ could have been bigger. First the approximation $\norm{t}
\omega \geq 2 \omega$ is very rough. Furthermore, picking a value $n > 3$
would have yielded a slightly bigger constant than $1/3$. However, we mostly
care about the linear relation between the statistical distance and the
divergence, so the current bound is good enough.

\begin{corollary}
  Assuming $\omega \geq \log(n)$ the statistical distance $\statdist(t R,
  I)$ can be bounded to the range
  \[
    \statdist(t R, I) \in \bracks{n \cdot \Theta(2^{-2 \omega}),
    \sqrt{n} \cdot \Theta(2^{-\omega})} \ .
  \]
\end{corollary}

\subsection{Divergence for sum of products}
As discussed in practice we would like to analyze the independence of the
coefficients of $t_1 R_1 + \ldots + t_s R_s$ for $s \geq 2$. If we want to apply
\cref{lem:divDependence}, we need to compute the entropy $H(t_1 R_1 + \ldots +
t_s R_s)$. Unfortunately, the entropy of this sum is quite hard to
compute in general. For
some related work on the entropy of sums see \cite{green2024}.

We can easily reason about one special case though. Suppose $t \eqdef t_1 =
\ldots = t_s$. Then $t_1 R_1 + \ldots + t_s R_s = t(R_1 + \ldots + R_s)$, is
invertible. As $R_1 + \ldots + R_s ~ \Ber(s \omega)^{\otimes n}$, we return to
our analysis of $t R$. Now the coefficients of $t (R_1 + \ldots + R_s)$ are
marginally distributed according to $\Ber(\weight{t} s \omega)$. On the other
hand $H(t (R_1 + \ldots + R_s)) = n \tilde h(s \omega)$. Therefore, for $I_s
\leftarrow \Ber(\weight{t} s \omega)^{\otimes n}$, we have \[ D(t(R_1
    + \ldots + R_s)
\| I_s) = n (\tilde h(\weight{t} s \omega) - \tilde h(s \omega)). \]

For the general case, it seems much harder to compute this entropy. We can use
the bound \[
  H(t_1 R_1 + \ldots + t_s R_s) \leq H(R_1) + \ldots + H(R_s) = s n
  \tilde h(\omega).
\] Unfortunately, this bound is (for parameters of interest) weaker
than the trivial upper bound \[
  H(t_1 R_1 + \ldots + t_s R_s) \leq n,
\] as we have a distribution over $n$ bits.

Instead, we can analyze the entropy based on the distribution of the
coefficients. If
we write
\[
  C_0 + C_1 X + \ldots + C_{n - 1} X^{n - 1} \eqdef t_1 R_1 + \ldots + t_s R_s,
\] then $H(t_1 R_1 + \ldots + t_s R_s) = H(C_0, C_1, \ldots, C_{n - 1})$. Recall
that for $\tau \eqdef \weight{t_1} + \ldots + \weight{t_s}$ the marginal
distributions are $C_0, \ldots, C_{n - 1} \leftarrow \Ber(\tau
\omega)$. Applying again the subadditivity of entropy to these variables
only yields the trivial upper bound
\[
  H(C_0, \ldots C_{n - 1}) \leq H(C_0) + \ldots + H(C_{n - 1}) = n
  \tilde h(\tau \omega).
\]
Recall that we are explicitly trying to compare the entropy of
$H(t_1R_1+\cdots+t_sR_s)$ to this value, so this bound yields nothing
of interest.

\medskip

Another possible approach is to compute the upper bound \[
  H(C_0, \ldots C_{n - 1}) \leq H(C_0, C_1) + H(C_2, C_3) + \dots + H(C_{n - 3},
  C_{n - 2}) + H(C_{n - 1}).
\] Recall, that we choose $n$ to be an odd prime, so ``splitting up''
the entropy
this way requires us to deal with $H(C_{n-1})$ separately. Note that
for each $1 \leq i < j \leq n-2$, $(C_i,
C_j)$ is distributed over 2 bits, and -- while this is a bit
cumbersome -- we can explicitly compute this
distribution. We will base our analysis on \cite{pacher2016}. This paper is
about the weight distribution of the syndrome under Bernoulli noise, i.e.
the distribution of $R T$ for $T \in \Ft^{s n \times n}$ a parity-check
matrix and $R \leftarrow \Ber(\omega)^{\otimes s n}$. Write
$T_i \in \Ft^{n \times n}$ for the cyclic matrix that represents multiplying
with $t_i(X)$. Then define \[
  T \eqdef
  \left(
    \begin{array}{ccc|ccc|ccc}
      &    & & &        & & &     & \\
      &T_1 & & & \ldots & & & T_s & \\
      &    & & &        & & &     &\\
  \end{array}\right)
\] which one can view as a parity-check matrix. Furthermore, define
$R$ for the concatenation
of $R_1, \ldots, R_s$ \[
  R = (R_1 | \ldots | R_s).
\] Then \[
  R T = t_1(X) R_1(X) + \ldots + t_s(X) R_s(X) = (C_0, \ldots, C_{n - 1}).
\]

It will now be important for us to view elements of $\F_2^n$, or $P_n$, as elements of $\R^n$, where we naturally map $0 \in \F_2$ to $0 \in \R$ and $1 \in \F_2$ to $1 \in \R$. For a vector $v \in \F_2^n$ (respectively, a polynomial $P \in P_n$), we denote by $\hat{v} \in \R^n$ (respectively, $\hat{P} \in \R^n$) the resulting values in Euclidean space. We also use the same notation for matrices.

The distribution of $(C_i, C_j)$ depends on the symmetric matrix
$\Lambda \in \R^{n \times n}$ defined as \[
  \Lambda \eqdef \hat T^\top \hat{T}
\]
Above, we emphasize that the multiplication is defined over the \textbf{real numbers}. Write
$\lambda_{ij}$ for the coefficients of $\Lambda$.
Based on the polynomials $t_1, \ldots, t_s$, the value of
$\lambda_{ij}$ is easy to compute. By definition, $\lambda_{ij}$
is the inner product (over the reals) of the $i$'th and $j$'th row
of the matrix $T$.
Equivalently \[ \lambda_{ij} = \ip{\widehat{X^i \cdot t_1}, \widehat{X^j \cdot t_1}}_\R +
\ldots + \ip{\widehat{X^i \cdot t_s}, \widehat{X^j \cdot t_s}}_\R,
\] where $\ip{\cdot, \cdot}_\R$ is defined as the usual inner product as
vector over the reals, and recall further that the hat notation means that we view the elements as lying in $\{0,1\}^n \subseteq \R^n$. As the multiplication by $X^i$ is just a
shift of the vector by $i$ positions, we can replace two shifts by just one
shift, so
\begin{align}
\lambda_{ij} = \ip{\hat t_1, \widehat{X^{i - j} \cdot t_1}}_\R + \ldots + \ip{\hat t_s, \widehat{X^{i - j}
		\cdot t_s}}_\R. \label{eq:shifts}
\end{align} Notably, the value of $\lambda_{ij}$ only depends on the
\emph{difference} between the coefficients $i - j$.

Alternatively, one can note that for each $\ell \in [s]$,
\[
\ip{\hat t_\ell, \widehat{X^{i - j} \cdot t_\ell}}_\R = |\mathrm{supp}(t_\ell)
\cap \mathrm{supp}(X^{i-j}t_\ell)| \ ,
\]
where for a polynomial $P(X) = \sum_{i=0}^{n-1}P_i X^i \in
\F_2[X]/(X^n-1)$ we have defined $\mathrm{supp}(P) = \{i \in [n] :
P_i =1\}$.

Recall that
$p(\omega) = (1 - 2^{-\omega}) / 2$. Following \cite[Eq. (14)]{pacher2016}, we can obtain the following joint distribution table for each $C_i$ and $C_j$:
\begin{table}[H]
  \begin{tabular}{l|ll}
    $C_i$ \textbackslash{} $C_j$ & 0
    & 1
    \\ \hline
    0                            & $1 - p(\tau \omega) -
    \frac{1}{2}p(2 (\tau - \lambda_{i j}) \omega)$ & $\frac{1}{2} p(2
    (\tau - \lambda_{i j}) \omega)$                  \\
    1                            & $\frac{1}{2} p(2 (\tau -
    \lambda_{i j}) \omega)$                     & $p(\tau \omega) -
    \frac{1}{2} p(2 (\tau - \lambda_{i j}) \omega)$
  \end{tabular}
  \caption{Distribution for $C_i, C_j$}
  \label{table:distr}
\end{table}
To give a sense of how these quantities can be determined, we provide an example computation. Recall that $(C_1,\dots,C_n)$ are determined by i.i.d. $\Ber(\omega)$ random variables $R_{u,t}$ for $u \in [n]$ and $t \in [s]$. Fix sets $S_i$ and $S_j$ such that $C_i = \sum_{(u,t) \in S_i}R_{u,t}$ (read modulo 2) and $C_j = \sum_{(u,t)\in S_j} R_{u,t}$. Then $|S_i| = |S_j| = \tau$ and $S_{ij}:=S_i \cap S_j$ has $|S_{ij}| = \lambda_{ij}$. Denote by $E_{ij}$ the event that $\sum_{(u,t) \in S_{ij}}R_{u,t}\equiv 0 \pmod{2}$, and note that conditioned on $E$ we have that $C_i=0$ iff $\sum_{(u,t)\in S_i \setminus S_j} R_{u,t}=0 \pmod{2}$, and similarly for the probability $C_j=1$. Similarly, conditioned on $E_{ij}=1$ we consider the event that this sum modulo 2 takes on the opposite value. Due to \Cref{lem:pilUp}, we therefore have
\begin{align*}
	\Pr[C_i = 0 \land C_j = 1] &= \Pr[C_i=0 \land C_j=1 | E_{ij}]\Pr[E_{ij}]\\
	&\qquad + \Pr[C_i=0\land C_j=1|\neg E_{ij}]\Pr[\neg E_{ij}]\\
	&= (1-p(\omega(\tau - \lambda_{ij})))p(\omega(t-\lambda_{ij}))(1-p(\omega\lambda_{ij})) \\
	&\qquad+p(\omega(\tau-\lambda_{ij})) (1-p(\omega(\tau-\lambda_{ij})))p(\omega\lambda_{ij}) \\
	&= \frac{1}{2}p(2\omega(\tau-\lambda_{ij})),
\end{align*}
where in the last line we used the identity $p(x)(1-p(x))=\frac12 p(2x)$.

This distribution (\Cref{table:distr}) yields the entropy
\begin{align*}
  H(C_i, C_j)
  = {} & H(C_i) + H(C_j \mid C_i) \\
  = {} & \tilde h(\tau \omega) + p(\tau \omega) H(C_j \mid C_i =
  1) + (1 - p(\tau \omega)) H(C_j \mid C_i = 0) \\
  = {} & \tilde h(\tau \omega) + p(\tau \omega) h\parens{\frac{p(2
  (\tau - \lambda_{i j}) \omega)}{2 p(\tau \omega)}} + (1 - p(\tau
  \omega)) h\parens{\frac{p(2
  (\tau - \lambda_{i j}) \omega)}{2 (1 - p(\tau \omega))}}. \\
  \leq {} &  \tilde h(\tau \omega) + h\parens{p(\tau \omega) \frac{p(2
    (\tau - \lambda_{i j}) \omega)}{2 p(\tau \omega)} + (1 - p(\tau
    \omega)) \frac{p(2
  (\tau - \lambda_{i j}) \omega)}{2 (1 - p(\tau \omega))}}. \\
  = {} & \tilde h(\tau \omega) + \tilde h(2
  (\tau - \lambda_{i j}) \omega) \ ,
\end{align*}
where the inequality applies the fact that $h(\cdot)$ is concave.
Let's compare this to the entropy of $H(C_i) + H(C_j) = 2 \tilde
h(\tau \omega)$ for the idealized case where $C_i$ and $C_j$ are
independent. Note that if
$\lambda_{i j} > \tau / 2$, then the above establishes \[
  H(C_i, C_j) \leq \tilde h (\tau \omega) + \tilde h(2(\tau -
  \lambda_{ij})\omega) < 2
  \tilde h(\tau \omega) = H(C_i) + H(C_j).
\]

Thus, to make thisgap large we would like $\lambda_{ij}$ to be large. Recalling \eqref{eq:shifts}, to have $\lambda_{ij}$ large we need pairs $i,j$
for which shifting by $j-i$ leads to another coefficient vector with
large overlap. We now consider a specific case where we have such large overlap.
\paragraph{Arithmetic progression} If the support of the noise vectors is
sufficiently structured, then we can find cases such that the
$\lambda_{ij}$'s are big. By big we mean
$\lambda_{ij} \approx \tau$. Suppose that the support of the noise vector is an
arithmetic progression: so \[
  \text{supp}(t_k) = \set{a x + b_k \mid x \in \{0, \ldots, \weight{t_k} - 1\}}
\] for all $k \in [s]$. Then \[
  \ip{\hat t_k, \widehat{X^a t_k}}_\R = \abs{\text{supp}(t_k) \cap \text{supp}(X^a
  \, t_k)} \geq
  \weight{t_k} - 1.
\] Thus, for each $i$, we have \[
  \lambda_{i,i + a} = \ip{\hat t_1, \widehat{X^a t_1}}_\R + \ldots + \ip{\hat t_s, \widehat{X^a
  t_s}}_\R \geq \tau
  - s,
\] yielding \[
  H(C_i, C_{i + a}) \leq \tilde h(\tau \omega) + \tilde h(2 s \omega).
\]
If $n$ is odd, and $a$ and $n$ are coprime, then we can upper bound
the total entropy by
\begin{align*}
  H(C_0, \ldots, C_{n - 1})
  \leq{} & H(C_0, C_{a}) + H(C_{2a}, C_{3a}) + \ldots + H(C_{(n - 3)
  a}, C_{(n - 2) a})
  + H(C_{(n - 1) a})
  \\ \leq{} & \frac{n - 1}{2} \parens{\vphantom{\frac{n -
  1}{2}}\tilde h(\tau \omega) + \tilde h(2s \omega)} + \tilde h(\tau \omega)
\end{align*}
where the indices $i$ in $C_i$ are read modulo $n$, and using that
$0, a, \ldots,
(n - 1) a$ are distinct values mod $n$ when $a$ and $n$ are coprime.

In total using \cref{lem:divDependence} we get
\begin{align*}
  D(t_1 R_1 + \ldots + t_s R_s \| \Ber(\tau \omega)^{\otimes n})
  = {} & H(\Ber(\tau \omega)^{\otimes n}) - H(C_0, \ldots, C_{n - 1})
  \\ \geq {} & \frac{n - 1}{2} \parens{\vphantom{\frac{n - 1}{2}}
  \tilde h(\tau \omega) - \tilde h(2s \omega)}.
\end{align*}
Using the approximation for $\tilde h(\omega)$, we can characterize the
divergence asymptotically.

\begin{corollary}
  Let $\omega \leq \log(n)$, and $\tau \geq 2 s$, and $t_1(X), \ldots,
  t_s(X)$ be vectors that have an arithmetic progression with the same common
  difference. Then the divergence to $I \eqdef \Ber(\tau \omega)^{\otimes n}$
  can be bounded to the range \[
    D(t_1 R_1 + \ldots + t_s R_s, I) = \Theta(n 2^{-4 s \omega}).
  \] Notably, this value is negligible when $s \omega \geq \log(n)^2$.
\end{corollary}
\begin{proof}
    The proof is a generalization of \cref{cor:tRNegl}. Using the approximation
    $\tilde h = 1 - 2^{-2 \omega} + \O(2^{-4\omega})$ (see
    \cref{lem:approxh}), we know that \begin{align*}
        D(t_1 R_1 + \ldots + t_s R_s, I)
           & =    H(I) - H(t_1 R_1 + \ldots + t_s R_s)
        \\ & \geq \frac{n - 1}{2} (\tilde h(\tau \omega) - \tilde h(2 s \omega))
        \\ & =    \frac{n - 1}{2} \parens{\Theta(2^{-4 s \omega}) - \Theta(2^{-\tau \omega})}
        \\ & =    \Theta(n 2^{-4 s \omega}),
    \end{align*}
    as required.
\end{proof}

\paragraph{Statistical distance} Similar to the case $s = 1$ (see
\cref{subsec:stat-dist}), we can use the reverse Pinsker inequality to relate
the KL-divergence to the statistical distance. To achieve the bound we proof the
generalization of \cref{thm:noiseStatDistLower}.

\begin{theorem}  \label{thm:noiseStatDistLowerGeneric}
  Let $t_1, \ldots, t_s \in \Ft[X]/(X^n - 1)$ polynomials span the entire
  space, i.e. $\langle t_1, \ldots, t_s\rangle = \Ft[X]/(X^n - 1)$,
  and $R_1, \ldots, R_s \leftarrow
  \Ber(\omega)^{\otimes n}$ i.i.d. Furthermore, let $I~\leftarrow~\Ber(\tau
  \omega)^{\otimes n}$. For $n \geq 3$ and $\omega \geq \log(n)$, the
  statistical distance has a lower bound \[
    \statdist(t_1 R_1 + \ldots + t_s R_s \| I) \geq
    1 /(3s) \cdot D(t_1 R_1 + \cdots + t_sR_s \| I).
  \]
\end{theorem}
\begin{proof}

  In the proof we bound $\Pr[t_1 R_1 + \ldots + t_s R_s = x] / \Pr[I = x] \in
  [1/8^s, 8^s]$. Then, we apply \cref{thm:RevPinsker}. We can achieve these
  bounds, finding lower/upper bound on the enumerator and the denominator
  separately. Both probabilities can be bounded using the distributions
  $\Ber(\omega)$, and $\Ber(\weight{t} \omega)$. For any $x \in \Ft[X]/(X^n -
  1)$ there are at exactly $2^{sn - n}$ possible tuples $(r_1, \ldots, r_s)$
  such that $t_1 R_1 + \ldots + t_s R_s = x$. Then, we can bound the
  probability of each tuple $\Pr[(R_1, \ldots, R_s) = (r_1, \ldots, r_s)] \in
  [((1 - 2^{-\omega})/2)^{s n}, ((1 + 2^{-\omega})/2)^{s n}]$.
  \begin{align*}
    \frac{\Pr[t_1 R_1 + \ldots + t_s R_s = x]}{\Pr[I = x]}
    & \leq    \frac{\max_{x \in \Pn} \Pr[t_1 R_1 + \ldots + t_s R_s =
    x]}{\min_{x \in \Pn} \Pr[I = x]}
    \\ & \leq \frac{2^{s n - n} ((1 + 2^{-\omega})/2)^{sn}}{((1 -
    2^{-\weight{t} \omega})/2)^n}
    \\ & =    \frac{(1 + 2^{-\omega})^{sn}}{(1 - 2^{-\weight{t} \omega})^n}
    \\ & \leq    \frac{(1 + 2^{-\omega})^{sn}}{(1 - 2^{-\weight{t}
    \omega})^{sn}}
    \\ & \leq   \frac{(1 + 2^{-\omega})^{sn}}{(1 - 2^{-\omega})^{sn}}
    && \text{(using $\norm{t} \geq 1$)}
    \\ & \leq \parens{\frac{1 + 1/n}{1 - 1/n}}^{ns} && \text{(using
    $\omega \geq \log(n)$)},
  \end{align*}
  where we use that $\omega \geq \log(n)$. Define $u_n \eqdef \frac{(1 +
  1/n)^n}{(1 - 1/n)^n}$. The upper bound is to this ratio is equal to
  $(u_n)^s$. We claim
  that $u_n$ is decreasing, and $u_3 \leq 8$, which yields an upper bound
  of $8^s$.

  To prove this claim we compute the derivative of $u_n$. As $\ln(\cdot)$
  is monotonically increasing, if $f(x) = \ln\parens{\parens{\frac{1 +
  1/x}{1 - 1/x}}^x}$ has a negative derivative, then $u_n$ is indeed
  decreasing. For $x > 1$ we can give an upper bound for derivative by \[
    f'(x) = \ln\parens{\frac{x + 1}{x - 1}} -\frac{2x}{x^2 - 1} <
    \ln\parens{1 + \frac{1}{x}} - \frac{2}{x} \leq \frac{1}{x} - \frac{2}{x}
  = - \frac{1}{x} < 0 \] which is indeed negative.

  Using that $u_n$ is decreasing, filling in a value for $n = 3$ gives
  $u_3 = 8$ so \[
    \frac{\Pr[t_1 R_1 + \ldots + t_s R_s = x]}{\Pr[I = x]}
    \leq (u_n)^s \leq (u_3)^s = 8^s.
  \] Notably, the upper bound on this ratio independent of $n$.

  For a lower bound on the probability ratio we can apply a
  similar approach
  \begin{align*}
    \frac{\Pr[t_1 R_1 + \ldots + t_s R_s = x]}{\Pr[I = x]}
    & \geq \frac{\min_{x \in P_s} \Pr[t_1 R_1 + \ldots + t_s R_s =
    x]}{\max_{x \in \Pn} \Pr[I = x]}
    \\ & \geq  \frac{2^{sn - n} ((1 - 2^{-\omega})/2)^{s n}}{((1 +
    2^{-\weight{t} \omega})/2)^n}
    \\ & =     \frac{(1 - 2^{-\omega})^{sn}}{((1 + 2^{-\weight{t} \omega}))^n}
    \\ & \geq  \frac{((1 - 2^{-\omega}))^{sn}}{((1 + 2^{-\weight{t}
    \omega}))^{sn}}
    \\ & \geq \parens{\frac{\parens{1 - 1/n}^{n}}{\parens{1 +
    1/n}^n}}^s && \text{(using $\omega \geq \log(n)$)}
    \\ & = (1/u_n)^s
    \\ & \geq (1/8)^s
  \end{align*}

  Where we use that $1/u_n$ is monotonically increasing. In total, we
  achieve the required bounds \[
    (1/8)^s \leq \frac{\Pr[t_1 R_1 + \ldots + t_s R_s = x]}{\Pr[I =
    x]} \leq 8^s.
  \]

  To finish the proof we need to apply \cref{thm:RevPinsker}. Note
  that for $x > 0$ \[
    \frac{d}{dx} \parens{\frac{x \log(x)}{x - 1}} = \ln(2) \cdot \frac{x
    - \log(x) - 1}{(x - 1)^2} \geq 0
  \]

  Because the derivative is strictly positive, $x \log(x) /(x - 1)$ is
  monotonically increasing. On the other hand $x \log(x) /(1 - x)$ is
  monotonically decreasing. Therefore, using our bound $M \leq 8$, and $m
  \geq 1/8$, we can give an upper bound on quantity from
  \cref{thm:RevPinsker}
  \begin{align*}
    \frac{m \log(m)}{1 - m} + \frac{M \log(M)}{M - 1}
    & \leq \frac{(1/8^s) \cdot \log(1/8^s)}{1 - (1/8^s)} + \frac{8^s
    \log(8^s)}{8^s - 1}
    \\ & = 3 s \Bigl(\dfrac{8^{s}}{8^{s}-1}-\dfrac{8^{-s}}{1-8^{-s}}\Bigr)
    \\ & = 3s.
  \end{align*}
  We can conclude by \cref{thm:RevPinsker} that \[
    D(t_1 R_1 + \ldots + t_s R_s \| I) \leq 3s \cdot \statdist(t_1
    R_1 + \ldots + t_s R_s, I),
  \] so \[
    \statdist(t_1 R_1 + \ldots + t_s R_s, I) \geq 1/(3s) \cdot D(t_1
    R_1 + \ldots + t_s R_s \| I).
  \] The theorem follows.
\end{proof}

 \bmhead{Acknowledgements} The authors would like to thank Jeroen
Zuiddam and Christian Schaffner for helpful feedback. MB's work was
supported by NWO Gravitation Project QSC. NR is supported by an NWO
(Dutch Research Council) grant with number C.2324.0590. This work was
done in part while NR was visiting the Simons Institute for the Theory
of Computing, supported by DOE grant \# DE-SC0024124.

\printbibliography

\end{document}